\pdfoutput=1
\documentclass[journal]{IEEEtran}

\usepackage[utf8]{inputenc}
\usepackage[T1]{fontenc}
\usepackage{amsmath,amssymb,amsthm}
\usepackage{mathtools}
\usepackage{bm}
\usepackage{graphicx}
\usepackage{booktabs}
\usepackage{enumitem}
\usepackage{xcolor}
\usepackage{algorithm}
\usepackage{algorithmic}
\usepackage[hidelinks]{hyperref}
\usepackage{revcref}
\usepackage{pdfpages}

\theoremstyle{plain}
\newtheorem{theorem}{Theorem}

\theoremstyle{definition}

\theoremstyle{remark}
\newtheorem{remark}{Remark}

\newcommand{\Phihat}{\hat{\bm{\Phi}}}
\newcommand{\Gam}{\bm{\Gamma}}
\newcommand{\X}{\mathbf{X}}
\newcommand{\D}{\mathbf{D}}
\newcommand{\W}{\mathbf{W}}
\newcommand{\G}{\mathbf{G}}
\newcommand{\I}{\mathbf{I}}
\newcommand{\M}{\mathbf{M}}
\newcommand{\SigW}{\bm{\Sigma}_W}
\newcommand{\Smat}{\mathbf{S}}
\newcommand{\tr}{\mathrm{Tr}}

\newcommand{\argmin}{\operatorname{argmin}}
\newcommand{\hsp}[3]{\hspace{#1pt}#3\hspace{#2pt}}

\newlength{\myfigwidth}
\newlength{\doublefigwidth}
\begin{document}
\bstctlcite{IEEEexample:BSTcontrol}

\title{Learning Regularization Structure for Biosignal Template Estimation}

\author{Yonathan Murin, Alexandre Gramfort%
\thanks{The authors are with Reality Labs, Meta Platforms, Inc (e-mail: moriny@meta.com; agramfort@meta.com). Code to reproduce all experiments is available at \url{https://github.com/facebookresearch/Learning_Regularization_Structure_for_Biosignal_Template_Estimation}.}
\thanks{This is the authors' accepted version of an article accepted for publication in IEEE Transactions on Biomedical Engineering, DOI: \href{https://doi.org/10.1109/TBME.2026.3757665}{10.1109/TBME.2026.3757665}.}
\thanks{\copyright~2026 IEEE. Personal use of this material is permitted. Permission from IEEE must be obtained for all other uses, in any current or future media, including reprinting/republishing this material for advertising or promotional purposes, creating new collective works, for resale or redistribution to servers or lists, or reuse of any copyrighted component of this work in other works.}
\vspace{-4mm}}
\date{}

\maketitle

\begin{abstract}
 \textit{Objective:} To estimate event-related templates of biosignals accurately from a small number of events. 
\textit{Methods:} We design a template estimation algorithm that minimizes Stein's unbiased risk estimate (SURE) over a convolution kernel defining a Tikhonov penalty, 
learning its structure and magnitude from the data, and extend SURE to temporally correlated noise. We compare the resulting estimator, SURE(kernel), 
with established estimators in simulations and on auditory EEG, P300 (ten subjects), and ECG recordings. 
\textit{Results:} In simulations, SURE(kernel) was the most accurate method at moderate and high noise levels and, at moderate noise, for every number of events tested; 
at low noise, generalized cross-validation was the most accurate. On real data with five events per class, 
SURE(kernel) was the closest of the compared estimators to the full-data references on the auditory EEG and in all ten P300 subjects. 
On ECG, averaging was closest; SURE(kernel) stayed within 3\% of it, whereas generalized cross-validation and Empirical Bayes were 38\% and 40\% farther. 
On P300, SURE(kernel)'s template-implied event probabilities were less overconfident than those of averaging, at similar detection performance. 
On both EEG datasets, absolute accuracy at five events remained limited. \textit{Conclusion:} Learning the regularization structure from the data, 
with a correction for correlated noise, can make templates estimated from few events more accurate than those of established estimators. 
\textit{Significance:} The method yields more accurate templates when few events can be recorded, as in personalization or when collecting data from each user is costly, 
with better-calibrated event probabilities at similar detection performance.
\end{abstract}

\begin{IEEEkeywords}
Template estimation, Tikhonov regularization, SURE, colored noise, event-related potentials, EEG, ECG, biosignal processing
\end{IEEEkeywords}

\vspace{-3mm}
\section{Introduction}
\label{sec:introduction}

Estimating event-locked templates from noisy biosignal recordings is a foundational challenge spanning multiple biomedical modalities,
including EEG~\cite{luck2014introduction}, MEG~\cite{hamalainen1993magnetoencephalography}, ECG~\cite{moody2001mitbih}, EMG~\cite{farina2004extraction,englehart2003robust,kaifosh2025}, 
and fMRI~\cite{dale1999optimal}.
Rapid calibration from few ($K$) events is desirable for personalization and bedside monitoring, so estimation quality in the low-data regime matters in practice.
The classical approach of signal averaging improves the signal-to-noise ratio (SNR) proportionally to $\sqrt{N}$~\cite{dawson1954summation},
but fails when short inter-stimulus intervals cause adjacent responses to overlap temporally, introducing systematic bias~\cite{woldorff1993adjacent}.
While single-trial denoising frameworks address trial variability~\cite{woody1967characterization,quian2003single,georgiadis2005single,davanzo2011bayesian},
they fundamentally do not resolve this overlap problem.

To handle overlapping events, an alternative approach models the continuous recording as a linear superposition of event-locked templates: $\X \hsp{-2}{-2}{=} \D \bm{\Phi} + \W$,
where $\D$ is a design matrix encoding event timings, $\bm{\Phi}$ contains the target templates, and $\W$ represents noise~\cite{friston1994statistical,dale1997selective,dale1999optimal,smith2015rerpI,smith2015rerpII,ehinger2019unfold}.
The least-squares solution deconvolves overlapping contributions and is unbiased, but it can have a large variance with few events, 
particularly when temporal overlap makes the design ill-conditioned. 
Tikhonov regularization~\cite{tikhonov1977solutions} accepts a small bias for a larger reduction in variance: for a suitable penalty magnitude, 
the regularized estimate has a smaller mean squared error than least squares, as shown for ridge regression~\cite{hoerl1970ridge}. 
It then requires choosing the regularization structure (e.g., smoothness penalty, ridge, or a combination) and its magnitude, 
which are typically chosen heuristically; this choice matters most in the low-data regime.

Template estimation via regularized least squares is a linear inverse problem with several classical parameter selection strategies,
including GCV~\cite{golub1979gcv}, the L-curve method~\cite{hansen1992lcurve}, and Morozov's discrepancy principle~\cite{morozov1984methods}.
Alternatively, Stein's Unbiased Risk Estimate (SURE)~\cite{stein1981sure} targets the mean-square error (MSE) directly,
offering a powerful framework for data-driven optimization~\cite{donoho1995adapting,eldar2009gsure,ramani2008montecarlo,murin2026surede}.
Two limitations, however, restrict the applicability of these methods to biosignal template estimation.
First, in the implementations compared here, these methods (including Empirical Bayes frameworks~\cite{mackay1992bayesian,friston2007variational,friston2008multiple,trujillo2004bayesian})
select only a scalar magnitude for an a-priori chosen regularization operator, 
leaving the regularization structure (smoothness, ridge, or a combination) as a heuristic choice; the frameworks can tune more parameters in principle.
Second, standard SURE assumes independent white noise, an assumption routinely violated by the temporally correlated,
$1/f$-type noise ubiquitous in bandpass-filtered biosignals~\cite{pfurtscheller1999event}.

 In the closely related general linear model for fMRI, colored residuals are handled by pre-whitening the data and the design matrix with an autoregressive (AR) noise model~\cite{friston1995analysis,worsley1995analysis}, which changes the fitting metric from ordinary to weighted least squares. 
We instead leave the data and design matrix untouched and correct the risk estimate. 
he correction is exact for a known noise covariance and is computed through the compact $P \hsp{-1}{-1}{\times} P$ matrix $\D^\top\SigW\D$ 
rather than a $T \hsp{-1}{-1}{\times} T$ covariance. In practice, we model the noise of the EEG recordings as AR(1), 
whose two parameters can be estimated from the short pre-event baselines of the few selected events; 
with five events per class, our estimator with a learned kernel was closer to the full-data references with this simple model than with a third-order autoregressive model 
estimated from the same baselines, in every recording tested (\Cref{sec:discuss_noise}).

In this work, we develop a data-driven framework for optimizing Tikhonov regularization in event template estimation. Our contributions are as follows:
\begin{enumerate}
\item We derive a SURE-based objective to optimize the regularization \emph{kernel structure, not just its magnitude}, by parameterizing the operator as a convolution matrix. 
The learned kernels can combine smoothness enforcement with ridge-like shrinkage, which scaling a fixed difference operator cannot.
\item We extend SURE to colored noise via a structured covariance correction ($\D^\top\SigW\D$) that accounts for event timing
geometry and is exact for a known noise covariance.
\item We evaluate the method in simulations and on auditory event-related potential (ERP), P300 brain--computer interface (BCI), 
and ECG data against established estimators and a recent comparator at matched numbers of events, with trivial-template baselines and a held-out detection endpoint.
\end{enumerate}

\Cref{sec:formulation,sec:method} present the model and the method, \Cref{sec:simulation} the simulations, \Cref{sec:real_data} the real-data evaluation, 
\Cref{sec:discussion} the discussion, and \Cref{sec:conclusion} the conclusion.

\vspace{-3mm}
\section{Problem Formulation}
\label{sec:formulation}

\textbf{Notation.}
Boldface uppercase letters ($\X$, $\D$) denote matrices, boldface lowercase ($\bm{h}$, $\bm{x}$) denote vectors, and regular italic ($T$, $C$, $\sigma$) denote scalars. 
$\|\cdot\|_F$ is the Frobenius norm, $\|\cdot\|$ the Euclidean norm of a vector, $\tr(\cdot)$ the matrix trace, $\otimes$ the Kronecker product, and $(\cdot)^\top$ the transpose. 
We assume that $\D$ has full column rank, so that its Moore-Penrose pseudoinverse is $\D^\dagger \hsp{-2}{-2}{=} (\D^\top\D)^{-1}\D^\top$.
$\I$ denotes the identity matrix of appropriate dimension. $\mathbb{E}\{\cdot\}$ denotes expectation over the noise.

\vspace{-3mm}
\subsection{Signal Model}

We consider a multi-dimensional linear estimation problem of the form:
\begin{equation}
\X = \D \bm{\Phi} + \W,
\label{eq:model}
\end{equation}
where $\X \hsp{-2}{-2}{\in} \mathbb{R}^{T \times C}$ is the observed multichannel signal with $T$ time samples and $C$ channels,
$\D \hsp{-2}{-2}{\in} \mathbb{R}^{T \times P}$ is the design matrix encoding the timing of $M$ event types with template length of $L$ samples 
(so $P \hsp{-2}{-2}{=} ML$), $\bm{\Phi} \hsp{-2}{-2}{\in} \mathbb{R}^{P \times C}$ contains the $M$ templates to be estimated,
and $\W \hsp{-2}{-2}{\in} \mathbb{R}^{T \times C}$ denotes measurement noise.
Each column of $\W$ is assumed to have zero mean and temporal covariance $\SigW \hsp{-2}{-2}{\in} \mathbb{R}^{T \times T}$,
and for the white-noise special case $\SigW \hsp{-2}{-2}{=} \sigma^2 \I_T$.
We do not assume a specific distribution for the noise, and the noise can be correlated across channels.

The design matrix $\D$ is constructed from the event occurrences:
for each event type $m$ with $N_m$ occurrences whose template windows start at $\{t_1^{(m)}, \ldots, t_{N_m}^{(m)}\}$,
the corresponding block $\D_m \hsp{-2}{-2}{\in} \mathbb{R}^{T \times L}$ has entries $[\D_m]_{t_i^{(m)}+\ell, \ell} \hsp{-2}{-2}{=} 1$
for $\ell = 0, \ldots, L-1$, and zero elsewhere.
Only windows that lie entirely inside the recording are included.\footnote{For the real data, each window starts $L/2$ samples before the event onset.}
When two occurrences of the same type are separated by fewer than $L$ samples, their contributions overlap in $\D_m$, so that some rows contain ones from multiple events.
The full design matrix is formed by horizontal concatenation: $\D \hsp{-2}{-2}{=} [\D_1, \ldots, \D_M]$, and when events of different types also overlap temporally, 
the corresponding columns of $\D$ become correlated.

Our objective is to estimate the template matrix $\bm{\Phi}$ from the observed signal $\X$ and the known event structure encoded in $\D$.
This overlap can make $\D^\top\D$ ill-conditioned and motivates the use of regularization.

\vspace{-3mm}
\subsection{Ordinary Least Squares}

The OLS estimator minimizes the sum-of-squares of the reconstruction error:
\vspace{-1mm}
\begin{equation}
\Phihat_{\text{OLS}} \hsp{-1}{-1}{=} \argmin_{\bm{\Phi}} \| \X - \D \bm{\Phi} \|_F^2 \hsp{-1}{-1}{\triangleq} \D^\dagger \X.
\label{eq:ols}
\end{equation}
While unbiased, the OLS solution has high variance (see a detailed discussion in~\cite{eldar2009gsure}),
particularly when $\D^\top\D$ is ill-conditioned due to temporal overlap between events or insufficient data.

\vspace{-3mm}
\subsection{Tikhonov Regularization}

To achieve a better bias-variance tradeoff, one can add a Tikhonov regularization~\cite{tikhonov1977solutions} term:
\begin{equation}
\Phihat \hsp{-1}{-1}{=} \argmin_{\bm{\Phi}} \| \X - \D \bm{\Phi} \|_F^2 + \| \Gam \bm{\Phi} \|_F^2,
\label{eq:tikhonov}
\end{equation}
which has the closed-form solution:
\begin{equation}
\Phihat \hsp{-1}{-1}{=} (\D^\top \D + \Gam^\top\Gam)^{-1} \D^\top \X.
\label{eq:tikhonov_solution}
\end{equation}

The matrix $\Gam$ encodes the desired regularization structure.
 We call $\Gam$ the regularization operator and $\Smat \hsp{-2}{-2}{=} \Gam^\top\Gam$ the penalty matrix.
All operators in this paper act on each template waveform separately and identically for all event types and channels:
$\Gam \hsp{-2}{-2}{=} \I_M \otimes \Gam_1$, where $\Gam_1$ has $L$ columns and acts on one template of one channel.
Setting $\Gam_1 \hsp{-2}{-2}{=} \alpha \I_L$ yields ridge regression.
Setting $\Smat \hsp{-2}{-2}{=} \lambda \Smat_2$ with $\Smat_2 \hsp{-2}{-2}{=} \I_M \otimes \bm{\Delta}^\top\bm{\Delta}$, 
where $\bm{\Delta} \hsp{-2}{-2}{\in} \mathbb{R}^{(L-2)\times L}$ is the second-order difference operator (rows $[1,-2,1]$), enforces smoothness;
here the scalar $\lambda \hsp{-2}{-2}{\geq} 0$ multiplies the penalty matrix.
In practice, both the structure and magnitude of $\Gam$ are chosen heuristically, which motivates the present work.

\vspace{-2mm}
\section{The Proposed Method}
\label{sec:method}

\subsection{Unbiased MSE Estimator (White Noise)}
\label{sec:sure}

We seek to optimize the regularization operator $\Gam$ to minimize the mean-square estimation error. 
Let $\bm{x}_c$, $\bm{\phi}_c$, and $\hat{\bm{\phi}}_c$ denote the $c$-th columns of $\X$, $\bm{\Phi}$, and $\Phihat$, respectively. 
For each channel $c$, the regularized estimate is $\hat{\bm{\phi}}_c = \G \bm{x}_c \in \mathbb{R}^P$, where
\begin{equation}
\label{eq:G_def}
    \G(\Gam) = (\D^\top\D + \Smat)^{-1} \D^\top \in \mathbb{R}^{P \times T}
\end{equation}
In the following, to simplify notation, we use $\G$ instead of $\G(\Gam)$, and $\bm{\theta}_c \hsp{-2}{-2}{=} \D^\dagger \bm{x}_c$ denotes the OLS estimate.
The per-column MSE is 
$\mathcal{J}_c \hsp{-2}{-2}{=} \mathbb{E}\{\|\bm{\phi}_c \hsp{-1}{-1}{-} \hat{\bm{\phi}}_c\|^2\}$.
Since $\mathcal{J}_c$ depends on the unknown $\bm{\phi}_c$, it cannot be computed directly.
Instead, we derive an unbiased estimator $\mathcal{J}'_c$ that depends only on observable quantities and the noise variance $\sigma^2$.
\begin{theorem}[White-noise SURE]
\label{thm:sure_white}
Let $\D$ have full column rank, let $\Smat$ be fixed, and let $\bm{w}_c$ have zero mean and covariance $\sigma^2 \I_T$. Then
\begin{equation}
\mathcal{J}'_c = \|\hat{\bm{\phi}}_c - \bm{\theta}_c\|^2 + 2\sigma^2 \tr(\G \D^{\dagger\top}) - \sigma^2 \tr\big((\D^\top\D)^{-1}\big)
\label{eq:sure_white}
\end{equation}
satisfies $\mathbb{E}\{\mathcal{J}'_c\} \hsp{-2}{-2}{=} \mathcal{J}_c$, 
with $\G \D^{\dagger\top} \hsp{-2}{-2}{=} (\D^\top\D + \Smat)^{-1} \hsp{-1}{-1}{\in} \mathbb{R}^{P \times P}$.
\end{theorem}

\begin{proof}
See supplementary material for the full proof.
\end{proof}

 Expanding the first term gives 
$\mathcal{J}'_c \hsp{-1}{-1}{=} \|\hat{\bm{\phi}}_c\|^2 - 2\hat{\bm{\phi}}_c^\top \bm{\theta}_c + 2\sigma^2 \tr(\G \D^{\dagger\top}) + \|\bm{\theta}_c\|^2 - \sigma^2 \tr((\D^\top\D)^{-1})$, and the last two terms do not depend on $\Gam$.
The trace term $\tr(\G \D^{\dagger\top})$ acts as a complexity penalty: it equals, after multiplication by $\sigma^2$, 
the sum of the covariances between corresponding components of the regularized and OLS estimates, and it decreases when a fixed penalty is scaled up.
This term is derived from the identity $\mathbb{E}\{\bm{w}^\top \mathbf{B} \bm{w}\} \hsp{-2}{-2}{=} \sigma^2 \tr(\mathbf{B})$, 
valid for \emph{any} matrix $\mathbf{B}$ when $\bm{w}$ has zero mean and covariance $\sigma^2\I$.
Note that in the case of temporally correlated noise this identity breaks, as discussed in the next section.

 The same $\Gam$ is applied to every channel, so we select it by minimizing the sum of the per-channel estimates. 
The minimizer of this sum need not minimize $\mathcal{J}'_c$ for each channel separately. 
Dropping the terms that do not depend on $\Gam$, the selected regularization is:
\vspace{-1mm}
 \begin{equation*}
\Gam_{\text{opt}} \hsp{-2}{-2}{=} \argmin_{\Gam} \hspace{-1pt} \sum_{c=1}^{C} \hspace{-1pt}
\left\{ \hspace{-1pt} \|\hat{\bm{\phi}}_c\|^2 \hsp{-1}{-1}{-} 2\hat{\bm{\phi}}_c^\top \bm{\theta}_c \hsp{-1}{-1}{+} 2\sigma^2 \tr(\G\D^{\dagger\top}) \hspace{-1pt} \right\} \hspace{-1pt}.
\end{equation*} 
Theorem~\ref{thm:sure_white} holds for each fixed candidate $\Gam$. As $\Gam_{\text{opt}}$ is selected from the data, 
and since $\sigma^2$ is estimated in practice (\Cref{sec:noise}), the minimized value is not an unbiased estimate of the risk of the selected estimator. 
We use this objective as a selection criterion.

\vspace{-3.5mm}
\subsection{Extension to Colored Noise}
\label{sec:colored}

Biosignal noise exhibits temporal autocorrelation, particularly after bandpass filtering.
When the noise covariance is different from $\sigma^2\I$, the white-noise SURE in Theorem~\ref{thm:sure_white} is biased.
The following theorem extends \Cref{thm:sure_white} to the case of colored-noise.
\begin{theorem}[Colored-noise SURE]
\label{thm:sure_colored}
 Let $\D$ have full column rank, let $\Smat$ be fixed, and let $\bm{w}_c$ have zero mean and known covariance $\SigW \hsp{-2}{-2}{\in} \mathbb{R}^{T \times T}$. Then
\begin{equation}
\mathcal{J}'_c = \|\hat{\bm{\phi}}_c - \bm{\theta}_c\|^2 + 2\tr(\M \D^\top\SigW\D) - \tr(\D^\dagger\SigW\D^{\dagger\top}),
\label{eq:sure_colored}
\end{equation}
satisfies $\mathbb{E}\{\mathcal{J}'_c\} \hsp{-2}{-2}{=} \mathcal{J}_c$, where $\M \hsp{-2}{-2}{=} (\D^\top\D)^{-1}(\D^\top\D + \Smat)^{-1} \in \mathbb{R}^{P \times P}$.
\end{theorem}

\begin{proof}[Proof sketch]
The result follows from \Cref{thm:sure_white} by replacing $\sigma^2\tr(\mathbf{A})$ with $\tr(\mathbf{A}\SigW)$, 
where $\mathbf{A} \hsp{-2}{-2}{=} \D^{\dagger\top}\G \in \mathbb{R}^{T\times T}$, and $\sigma^2\tr((\D^\top\D)^{-1})$ with $\tr(\D^\dagger\SigW\D^{\dagger\top})$. 
Direct evaluation in this $T \hsp{-1}{-1}{\times} T$ space is intractable ($T \hsp{-1}{-1}{\sim} 10^4$--$10^5$), 
but using $\D^{\dagger\top} \hsp{-2}{-2}{=} \D(\D^\top\D)^{-1}$ gives
\begin{equation}
\mathbf{A} \hsp{-2}{-2}{=} \D(\D^\top\D)^{-1}(\D^\top\D \hsp{-1.5}{-1.5}{+} \Smat)^{-1}\D^\top \hsp{-2}{-2}{=} \D\M\D^\top,
\label{eq:A_factored}
\end{equation}
where $\M \hsp{-2}{-2}{\in} \mathbb{R}^{P \times P}$. The cyclic trace property then yields
\begin{equation}
\tr(\D\,\M\,\D^\top\SigW) = \tr(\M\,\D^\top\SigW\D),
\label{eq:trace_reduction}
\end{equation}
reducing the $T \hsp{-1}{-1}{\times} T$ computation to a $P \hsp{-1}{-1}{\times} P$ trace. See supplementary material for the full derivation.
\end{proof}

Note that this result requires only that $\W$ has zero mean and known second-order statistics $\SigW$, no distributional assumption is needed.
Moreover, dropping the terms that do not depend on $\Gam$, the selected regularization minimizes the sum over all channels:
\begin{equation}
\hspace{-3pt} \Gam_{\hspace{-1pt} \text{opt}} \hsp{-2.5}{-2.5}{=} \argmin_{\Gam} \hspace{-3pt} \sum_{c=1}^{C} \hspace{-2.5pt} \left\{ \hspace{-2pt} \|\hat{\bm{\phi}}_c\|^2 \hsp{-2.5}{-2.5}{-} 2\hat{\bm{\phi}}_c^\top \hspace{-1pt} \bm{\theta}_c \hsp{-2.5}{-2.5}{+} 2\tr(\hspace{-1pt} \M \D^\top \hspace{-1pt} \SigW \hspace{-1pt} \D \hspace{-1pt}) \hspace{-2.5pt} \right\} \hspace{-3pt}.
\label{eq:optimization}
\end{equation}
As for \Cref{thm:sure_white}, the risk identity holds for each fixed candidate and a known $\SigW$.

Finally, when $\SigW \hsp{-2}{-2}{=} \sigma^2 \I_T$, we have $\D^\top\SigW\D \hsp{-2}{-2}{=} \sigma^2 \D^\top\D$, and the trace term reduces to
$2\sigma^2 \tr(\M\D^\top\D) \hsp{-2}{-2}{=} 2\sigma^2 \tr(\G\D^{\dagger\top})$, recovering Theorem~\ref{thm:sure_white}.

\begin{remark}[Structure of $\D^\top\SigW\D$]\label{rem:sigw_structure}
The matrix $\D^\top\SigW\D \in \mathbb{R}^{P\times P}$ has a concrete interpretation. 
For AR(1) noise with parameter $\varphi$ and marginal variance $\sigma^2$, we have $(\SigW)_{st} = \sigma^2 \varphi^{|s-t|}$, so:
\vspace{-1mm}
\begin{equation}
(\D^\top\SigW\D)_{(m,a),(n,b)} = \sum_{k=1}^{N_m}\sum_{l=1}^{N_n} \sigma^2\,\varphi^{|t_k^{(m)} + a - t_l^{(n)} - b|},
\label{eq:DtSigD}
\end{equation}
where $(m,a)$ indexes event type $m$ at template offset $a \in \{0,\ldots,L{-}1\}$, and the sum runs over all pairs of event occurrences.
Hence, this matrix captures event timing geometry:
\begin{itemize}
\item When the windows of events $k$ and $l$ are well separated (all lags $|t_k^{(m)} + a - t_l^{(n)} - b| \gg 1/|\!\log\varphi|$), 
the term $\varphi^{|\cdot|}$ is near zero, contributing negligibly.
\item When events are densely packed, many pairs contribute large values; these contributions change the covariance correction and hence the selected regularization.
\item When $\varphi \hsp{-1}{-2}{=} 0$ (white noise), $\D^\top\SigW\D \hsp{-2}{-2}{=} \sigma^2 \D^\top\D$, 
and Theorem~\ref{thm:sure_colored} reduces to Theorem~\ref{thm:sure_white}.
\end{itemize}
\end{remark}

\vspace{-3mm}
\subsection{Kernel Parameterization}
\label{sec:kernel}

\Cref{thm:sure_white} and \Cref{thm:sure_colored} hold for \emph{any} fixed regularization operator $\Gam$, 
which brings the question: \emph{Which regularization operator should be used?}
Since $\Gam$ acts independently on each template waveform (\Cref{sec:formulation}),
and each template is a one-dimensional signal along the temporal axis, we parameterize $\Gam$ as a convolution matrix defined by a kernel $\bm{h}$
 of length $Q$: $\Gam_1 \hsp{-2}{-2}{=} \mathbf{C}_{\bm{h}} \hsp{-2}{-2}{\in} \mathbb{R}^{(L-Q+1)\times L}$, whose rows are shifted copies of $\bm{h}$, 
so that $\mathbf{C}_{\bm{h}}\bm{\phi}$ is the part of the convolution that needs no samples outside the template.
Then, the optimizers solving for $\Gam_{\text{opt}}$ find not only the magnitude of the regularization but also its spectral shape.
To control the number of free parameters in our optimization problem we further impose the following constraints:
\begin{itemize}
\item \textbf{Kernel length} $Q$: controls the filter order (we use $Q \hsp{-2}{-2}{=} 3$; other lengths are compared in \Cref{sec:sim_results}).
\item \textbf{Symmetry}: constraining $\bm{h}$ to be symmetric leaves $(Q+1)/2$ free parameters for odd $Q$; 
with the kernel indexed about its center, its frequency response is real and even.
\end{itemize}

For example, a symmetric kernel of length 3 has the form $\bm{h} = [a, b, a]$ with 2 free parameters, compared to the fixed second-order difference $[1, -2, 1]$.
The optimization in~\eqref{eq:optimization} is performed over these kernel coefficients using 
the BFGS quasi-Newton method~\cite{nocedal2006} with the analytic gradient, computed through Cholesky factorizations of $\D^\top\D + \Smat$. 
The optimizer seeks a local minimum; we do not claim a global optimum.

The learned kernels are compared across datasets in \Cref{sec:learned_structure}. 
In the compared implementations, GCV, the L-curve, and Empirical Bayes select a scalar $\lambda$ for the fixed second-difference penalty; 
these frameworks can tune more parameters in principle, but the implementations used here do not.

\vspace{-3mm}
\subsection{Noise Parameter Estimation}
\label{sec:noise}

\Cref{thm:sure_colored} holds for any known covariance $\SigW$, but practical use requires a parametric model with few estimable parameters.
We approximate the noise by a stationary first-order autoregressive (AR(1)) process, 
whose covariance is fully determined by two scalars:
the marginal variance $\sigma^2$ and the lag-1 autocorrelation $\varphi$, with $(\SigW)_{st} \hsp{-2}{-2}{=} \sigma^2\varphi^{|s-t|}$.
AR(1) models have a long track record for capturing temporal autocorrelation in neuroimaging,
notably in fMRI where they were adopted to correct GLM statistics~\cite{friston1995analysis,worsley1995analysis}.
In the biosignal setting, the AR(1) model captures the dominant effect of bandpass filtering (temporal smoothing that correlates adjacent samples),
while remaining estimable from short pre-stimulus baseline segments.
Note that this is an approximation, discussed in \Cref{sec:discuss_noise}. Both parameters are estimated from the selected events only.

 \textbf{Baseline segments.} For each selected event, the baseline segment is the first $B \hsp{-2}{-2}{=} \lfloor 0.3L \rfloor$ samples of its template window, 
all before the event onset.

\textbf{Variance ($\sigma^2$).} We center each segment and channel on its own mean, 
compute $r_0 \hsp{-2}{-2}{=} B^{-1}\sum_t x_t^2$, and average over segments and channels.

\textbf{AR(1) coefficient ($\varphi$).} From the same centered segments, 
we compute $r_1/r_0$ with $r_1 \hsp{-2}{-2}{=} (B-1)^{-1}\sum_t x_t x_{t+1}$ per segment and channel, take the median, and clip it to $[0, 0.99]$. 
 The median estimate was $\varphi \hsp{-2}{-2}{=} 0.88$ in the central P300 analysis (range 0.82 to 0.93 over its 200 fits; 1--20\,Hz at 128\,Hz) 
and $0.98$ for the auditory ERP (1--40\,Hz at 600.6\,Hz).
Both estimates are approximations. Centering short segments biases the variance of correlated noise downward, and baselines can contain responses to neighboring events: 
in the P300 paradigm, onsets are 32 samples (250\,ms) apart, so no response-free baseline exists. 
When noise is approximately white by construction (e.g., simulation) or when baselines are contaminated by neighboring events (ECG),
we use \emph{pairwise} noise estimation: 
 $\hat{\sigma}^2$ is half the mean squared deviation of $\bm{x}_i - \bm{x}_j$ from its mean (over all samples and channels, divided by their number), 
averaged over all pairs of selected epochs of the same class. This cancels a template common to both epochs but ignores temporal correlation. 

Practical guidance on choosing between these estimators is given in Section~\ref{sec:discuss_noise}.

\begin{figure}[t]
\centering
\includegraphics[width=0.8\myfigwidth]{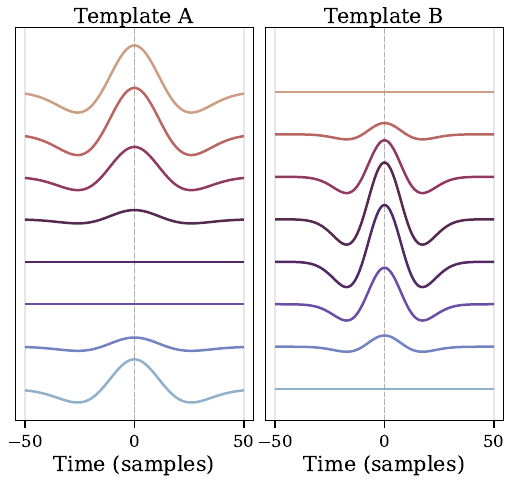}
\vspace{-2mm}
\caption{Synthetic templates used in the simulation. Each template is a Ricker wavelet ($w \hsp{-2}{-2}{=} 15$ for Template~A,
$w \hsp{-2}{-2}{=} 10$ for Template~B) modulated by a Hanning spatial profile across $C \hsp{-2}{-2}{=} 8$ channels
(centered for Template~B; shifted by 3 channels for Template~A).
Channels are vertically offset for visibility; same gap is used in both panels, so amplitudes are directly comparable across templates.}
\label{fig:templates}
\vspace{-3mm}
\end{figure}

\vspace{-2mm}
\section{Simulation Study}
\label{sec:simulation}

\vspace{-1mm}
\subsection{Setup}

We generate synthetic data following the model in~\eqref{eq:model}.
Two templates are constructed from Ricker wavelets of duration $L \hsp{-2}{-2}{=} 100$ samples with width parameters 15 and 10,
multiplied by a Hanning spatial profile across $C \hsp{-2}{-2}{=} 8$ channels (\Cref{fig:templates}).
Events are placed at random times with 
inter-event intervals equal to $m$ samples plus an exponentially distributed interval with mean $m$ samples, 
rounded down to whole samples (mean interval about $2m - 0.5$ samples; $m \hsp{-2}{-2}{=} 40$ unless stated otherwise), 
allowing temporal overlap. Noise is i.i.d.\ Gaussian ($\varphi \hsp{-2}{-2}{=} 0$ by construction), and noise variance is estimated using pairwise estimation.

The estimation quality is measured by the template estimation SNR, defined as the average over event types:
 $\text{SNR} \hsp{-2}{-2}{=} \frac{1}{M}\sum_{m=1}^{M} 10\log_{10}\frac{\|\bm{\phi}_m\|_F^2}{\|\bm{\phi}_m - \hat{\bm{\phi}}_m\|_F^2} \; [\text{dB}],$ 
where $\bm{\phi}_m$ and $\hat{\bm{\phi}}_m$ are the true and estimated templates for event type $m$, 
and norms are Frobenius (summing over all time samples and channels).

\vspace{-2mm}
\subsection{Competing Methods}
\label{sec:competing}

We compare the following template estimation approaches. All methods receive the same recording, design, and selected events:
\begin{itemize}
\item \textbf{Average}: simple epoch averaging of the selected windows of each class (overlap is ignored).
\item \textbf{OLS}: ordinary least squares on the linear model, \eqref{eq:ols}.
\item \textbf{GCV}: generalized cross-validation~\cite{golub1979gcv} for $\lambda$ selection with second-order differences
: $\lambda$ minimizes $T\|\X - \D\Phihat_\lambda\|_F^2 / (T - \tr(\D\G_\lambda))^2$ over 80 log-spaced values in $[10^{-3}, 10^{5}]$.
\item \textbf{L-curve}: 
 Hansen's maximum-curvature corner~\cite{hansen1992lcurve,hansen1993lcurve} of the curve 
$(\hat\rho, \hat\eta) \hsp{-2}{-2}{=} (\log\|\X - \D\Phihat_\lambda\|_F^2, \log\|(\I_M\otimes\bm{\Delta})\Phihat_\lambda\|_F^2)$, 
where the residual includes the whole recording. 
We select the $\lambda$ that maximizes the signed curvature $\kappa \hsp{-2}{-2}{=} (\hat\rho'\hat\eta'' - \hat\rho''\hat\eta')/(\hat\rho'^2 + \hat\eta'^2)^{3/2}$, 
with derivatives in $\log\lambda$ evaluated analytically. 
\item \textbf{Emp.\ Bayes}: evidence maximization~\cite{mackay1992bayesian} for $\lambda$ selection, with prior and noise precisions $\alpha$ and $\beta$ 
shared across channels and flat priors on the $q \hsp{-2}{-2}{=} 2M$ unpenalized components (constant and linear trend of each template).
\item \textbf{SURE($\lambda$)}: proposed method, optimizing $\lambda$ only with fixed $[1,-2,1]$ kernel. It minimizes \eqref{eq:optimization} over $\lambda \in [0, \infty]$. 
The candidates are the exact limits $\lambda \hsp{-2}{-2}{=} 0$ and $\lambda \hsp{-2}{-2}{=} \infty$, a log grid with 10 points per decade, 
refinements around every sampled local minimum, and the roots of the derivative in every grid interval where it changes sign from negative to positive.
\item \textbf{SURE(kernel)}: proposed method, optimizing the full kernel (symmetric, length~3; \Cref{sec:kernel}).
\item \textbf{Molina et al.} (real data only): the enhanced average of~\cite{molina2024enhanced}. 
For each class and channel, every selected epoch is aligned to the mean of the $K$ epochs by dynamic time warping, low-pass filtered, and averaged. 
We implemented the published procedure, with predetermined choices for the details it leaves open, and froze all settings before any real-data result was seen.
\end{itemize}

All regularized methods, except SURE(kernel) use the second-order difference structure for $\Gam$.
GCV, L-curve, and Emp.\ Bayes optimize $\lambda$ for this fixed structure.
SURE($\lambda$) optimizes $\lambda$ via SURE, while SURE(kernel) jointly optimizes the kernel coefficients.
SURE($\lambda$) and SURE(kernel) receive the same noise estimate.

\vspace{-3mm}
\subsection{Results}
\label{sec:sim_results}

\subsubsection{Influence of Noise Level}

 A noise-level sweep ($K \hsp{-2}{-2}{=} 20$, $\sigma \hsp{-2}{-2}{\in} \{0.005, 0.01, 0.02, 0.05, 0.1, 0.2, 0.5\}$, 20 realizations) shows that SURE(kernel) attains 
the highest SNR for $\sigma \hsp{-2}{-2}{\geq} 0.05$, by 0.65, 0.55, 0.49, and 0.21\,dB over Empirical Bayes at $\sigma \hsp{-2}{-2}{=} 0.05, 0.1, 0.2$, and $0.5$. 
At $\sigma \hsp{-2}{-2}{\leq} 0.02$, GCV has the highest SNR, and Empirical Bayes also exceeds SURE(kernel). 
In this regime the pairwise variance estimate is too large, because differences of overlapping epochs retain the responses to neighboring events: 
its median is 57 times the true noise variance at $\sigma \hsp{-2}{-2}{=} 0.005$, 4.3 times at $\sigma \hsp{-2}{-2}{=} 0.02$, and 1.13 times at $\sigma \hsp{-2}{-2}{=} 0.1$. 
Full curves are given in the supplementary material. 

\subsubsection{Influence of Number of Events}

\Cref{fig:nevents} shows the SNR versus the number of events per template type, with $\sigma \hsp{-2}{-2}{=} 0.1$ and $m \hsp{-2}{-2}{=} 50$.
SURE(kernel) has the highest SNR at every $K$, 0.48\,dB above Empirical Bayes at $K \hsp{-2}{-2}{=} 4$ (11.32 versus 10.84\,dB) 
and 0.98\,dB at $K \hsp{-2}{-2}{=} 100$ (23.67 versus 22.69\,dB).

\begin{figure*}[t]
\centering
\includegraphics[width=\doublefigwidth]{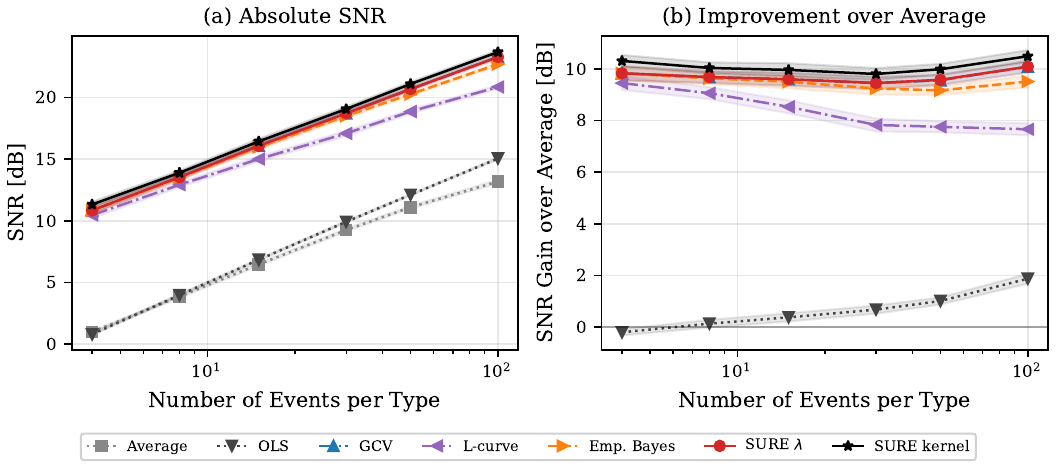}
\vspace{-2mm}
\caption{Simulation: SNR versus the number of events per type ($\sigma \hsp{-2}{-2}{=} 0.1$, $m \hsp{-2}{-2}{=} 50$). 
(a)~SNR. (b)~SNR gain over averaging. Bands: mean $\pm 1.96$ standard errors over 20 realizations (in (b), the two methods' standard errors combined). 
}
\label{fig:nevents}
\vspace{-3mm}
\end{figure*}

\subsubsection{Kernel Length and Symmetry}

We compared five kernel parameterizations: a ridge penalty ($Q \hsp{-2}{-2}{=} 1$) and kernels of length $Q \hsp{-2}{-2}{=} 3$ and 5, 
each symmetric or unconstrained, in a simulation (15 events per type, $\sigma \hsp{-2}{-2}{=} 0.1$, $m \hsp{-2}{-2}{=} 50$, 20 realizations; supplementary material). 
The $Q \hsp{-2}{-2}{=} 5$ fits were refitted with an iteration limit of 500, because many had stopped at the default limit of 50; all of them then converged. 
In the simulation, the symmetric and the unconstrained 3-tap kernels gave the same SNR (16.43\,dB), 
and the unconstrained kernel was nearly symmetric (largest relative asymmetry 0.0193); the converged 5-tap kernels were lower by 0.49\,dB (symmetric) and 0.73\,dB (unconstrained), 
and the ridge penalty reached only 7.72\,dB. We keep the symmetric 3-tap kernel for the rest of the work. 

\subsubsection{Influence of Noise Autocorrelation}

To separate the effect of learning the kernel from that of modeling colored noise, we crossed the fixed $[1,-2,1]$ kernel with 
SURE-optimized $\lambda$ (SURE($\lambda$)) and the learned kernel (SURE(kernel)) with the white-noise procedure 
(pairwise variance, \Cref{thm:sure_white}) and the colored procedure (baseline variance and estimated AR(1) coefficient, \Cref{thm:sure_colored}), 
on simulated AR(1) noise with $\varphi$ from 0 to 0.98, marginal standard deviation 0.1, $K \hsp{-2}{-2}{=} 15$ events per type, dense overlap ($m \hsp{-2}{-2}{=} 40$), 
and 20 realizations (\Cref{fig:colored_sim}). 
The learned kernel raised the SNR at every $\varphi$ under both procedures, by 0.07 to 0.40\,dB with the white procedure and by 0.23 to 0.54\,dB with the colored one. 
With the fixed kernel, the colored procedure helped most at intermediate $\varphi$ (by 0.51\,dB at $\varphi \hsp{-2}{-2}{=} 0.5$) and by 
at most 0.10\,dB at $\varphi \hsp{-2}{-2}{\geq} 0.85$. 
The sparse-overlap case, the interaction of the two factors, and the same comparison on P300 data are reported in the supplementary material. 

\begin{figure}[t]
\centering
\includegraphics[width=\myfigwidth]{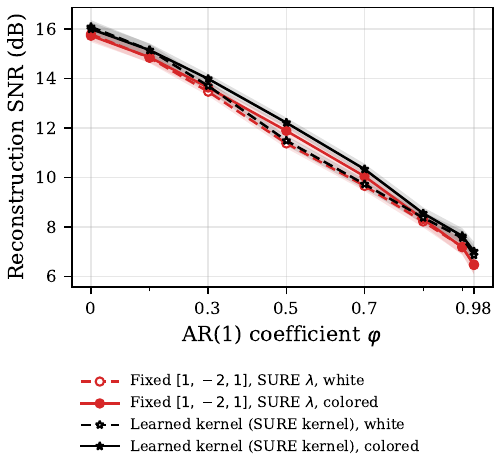}
\vspace{-2mm}
\caption{Simulation with AR(1) noise and dense overlap ($m \hsp{-2}{-2}{=} 40$, $K \hsp{-2}{-2}{=} 15$ events per type, 
nominal stationary noise standard deviation 0.1; the simulated noise starts nonstationary): SNR against the AR(1) coefficient $\varphi$. 
Bands: mean $\pm 1.96$ standard errors over 20 realizations, the same realizations for all four curves. 
}
\label{fig:colored_sim}
\vspace{-3mm}
\end{figure}

\vspace{-3mm}
\section{Real Data Analysis}
\label{sec:real_data}

We evaluate the proposed method on three bio-signal datasets. 
The evaluation protocol is common to all three (\Cref{sec:protocol}).
All datasets undergo standard dataset-specific pre-processing (band-pass filtering, and for 
the auditory EEG, exclusion of the channel marked as bad; no trials are rejected) prior to template estimation;
no additional featurization is applied.
Before reporting distance metrics, \Cref{fig:templates_real} shows the estimated waveforms themselves on two representative datasets, 
contrasting the slow P300 wave with the sharp ECG QRS complex.
On the P300 panels, the SURE(kernel) estimates vary least across selections, and their mean is shrunk toward zero. 
On the ECG Normal beat, the Empirical Bayes mean underestimates the R peak, whereas the variability across selections is similar for the three estimators; 
on the PVC beat the three estimators are close to the reference.

\vspace{-3mm}
\subsection{Evaluation Protocol}
\label{sec:protocol}

\textbf{Observation model.} $K$ denotes the number of labeled events selected at random per class.  The design $\D$ contains only the selected events. 
In the reference-distance analyses, $\X$ is the whole preprocessed recording (in the held-out evaluation of \Cref{sec:p300}, only the training blocks). 
Responses to the unselected events therefore remain in the data, overlapping the selected ones.

\textbf{References and distance.} The true templates are unknown. 
We therefore compare each low-$K$ estimate with reference templates computed from all events of the same recording: full-data OLS (primary, 
because it imposes no regularization structure), and full-data Empirical Bayes and the full-data average (sensitivity references). 
Each saved low-$K$ estimate is scored against all three references with the distance 
$d(\Phihat, \mathbf{R}) \hsp{-2}{-2}{=} \frac{1}{M}\sum_{m=1}^{M}\|\Phihat_m - \mathbf{R}_m\|_F$.

\textbf{Baselines.} Two templates that use none of the subject's selected events are evaluated with every real-data distance: 
the zero template (all zeros), and, for the P300 data, the population template, the class-wise mean of the full-data OLS references of the other nine subjects. 
We also report the $\lambda \hsp{-2}{-2}{=} \infty$ (straight-line) template of the second-difference penalty.

\textbf{Statistics.} For the ten P300 subjects, inference uses subjects as units, with repeated selections as technical replicates. 
The primary family compares SURE(kernel) with Empirical Bayes, SURE($\lambda$), and Average on the distance to the OLS reference at $K \hsp{-2}{-2}{=} 5$, 
by the two-sided exact Wilcoxon signed-rank test on subject-mean log distance ratios~\cite{wilcoxon1945}, with Holm correction over the three tests~\cite{holm1979}. 
Prespecified secondary families (other references, held-out AUC, and the comparator of~\cite{molina2024enhanced}) are Holm-corrected in the same way; 
all other comparisons, and all results for the single auditory EEG and ECG recordings, are descriptive (details in the supplementary material).

\begin{figure*}[t]
\centering
\includegraphics[width=\doublefigwidth]{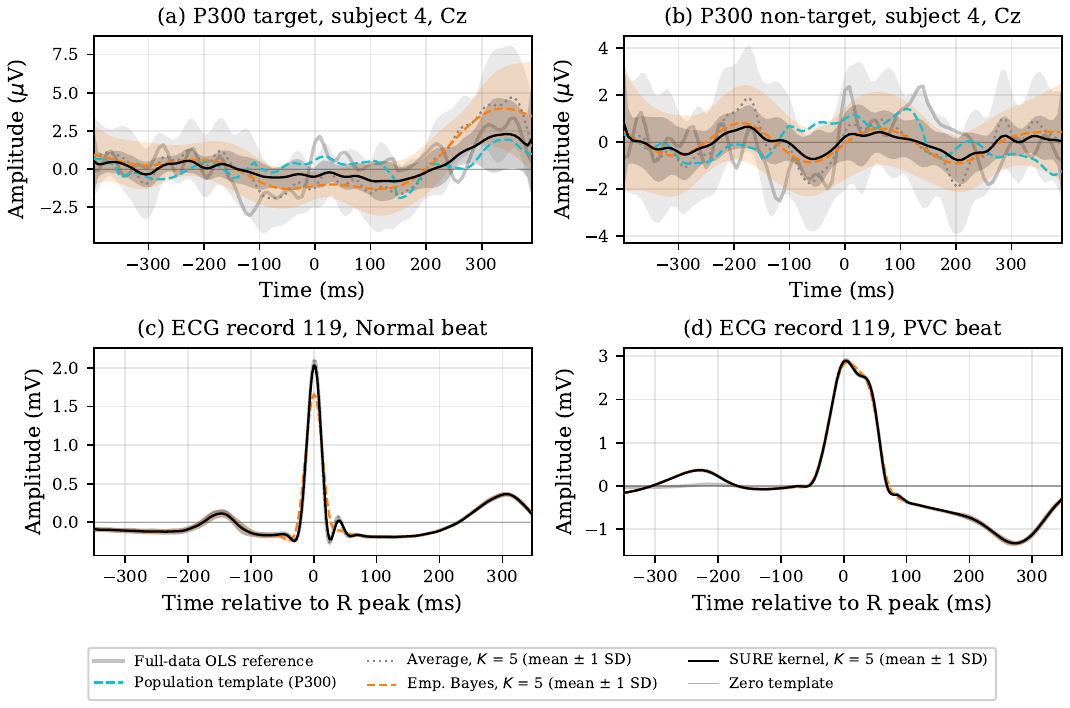}
\vspace{-2mm}
\caption{Estimated templates at $K \hsp{-2}{-2}{=} 5$: mean over 20 random selections $\pm 1$ standard deviation across them. 
Thick light gray: the full-data OLS reference; thin line at 0: the zero template. 
(a,~b)~P300, channel Cz, target and non-target classes ($\mu$V; time relative to stimulus onset). 
The subject was chosen by a fixed rule: the lower-median subject (rank 5 of 10) by the mean distance of SURE(kernel) to the OLS reference 
at $K \hsp{-2}{-2}{=} 5$ (subject 4, 55.0\,$\mu$V); it illustrates the estimators and is not a summary of the cohort, which is given in \Cref{sec:p300}. 
Cyan dashed: the population template (mean of the other nine subjects' full-data OLS references). 
(c,~d)~ECG record 119, Normal and PVC beats (mV; time relative to the R peak). The L-curve and the $\lambda \hsp{-2}{-2}{=} \infty$ template are not shown. 
}
\label{fig:templates_real}
\vspace{-3mm}
\end{figure*}

\vspace{-3mm}
\subsection{Auditory Event-Related Potentials (EEG)}
\label{sec:eeg}

We use the MNE sample dataset~\cite{gramfort2013mne}, which contains auditory and visual event-related potentials recorded with a 306-channel Neuromag system.
We select the 59 EEG channels that are not marked as bad and two auditory event types (left and right), 
providing 72 and 73 events respectively.
The data are bandpass-filtered to 1--40\,Hz at a sampling rate of 600.6\,Hz. Template duration is set to $\pm200$\,ms (240~samples).
The parameters of the noise model are estimated from pre-stimulus baselines with $\hat{\varphi} \hsp{-2}{-2}{\approx} 0.98$ 
(reflecting the narrow 1--40\,Hz passband relative to the 600.6\,Hz sampling rate), while we use the colored-noise SURE.
No two event windows overlap in this recording, so the full-data OLS and Average references coincide.

\Cref{fig:eeg_ecg}(a) and \Cref{tab:summary} show the distance to the full-data OLS reference. 
At $K \hsp{-2}{-2}{=} 5$, SURE(kernel) was the closest of the seven compared estimators, 28.5\% below Empirical Bayes and 35.3\% below Average; 
the order is the same under the EB reference (supplementary material). 
SURE($\lambda$), which uses the same colored-noise input, was farther than Empirical Bayes, so on this recording the learned kernel, 
not the noise model alone, separates SURE(kernel) from SURE($\lambda$). 
The zero template was closer to the reference than every estimator at $K \hsp{-2}{-2}{=} 5$. SURE(kernel) fell below it from $K \hsp{-2}{-2}{=} 10$, 
and all seven estimators from $K \hsp{-2}{-2}{=} 20$. At $K \hsp{-2}{-2}{=} 50$, SURE(kernel) was still the closest estimator to the OLS reference, 
whereas under the EB reference Empirical Bayes was closest. These are 20 random selections from one recording. 

\begin{figure*}[t]
\centering
\includegraphics[width=\doublefigwidth]{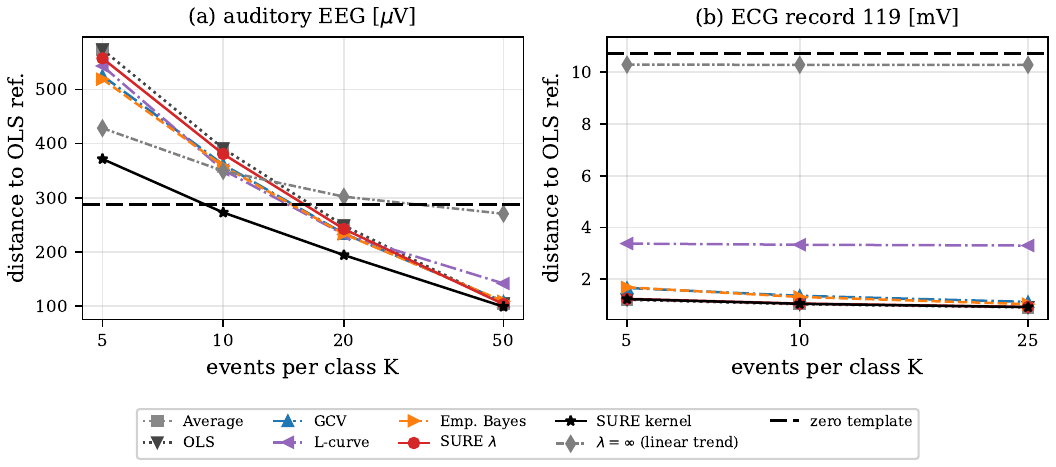}
\vspace{-2mm}
\caption{Auditory EEG and ECG: mean distance to the full-data OLS reference against the number of events per class $K$, 
over 20 random selections from one recording. (a)~Auditory EEG ($\mu$V). (b)~ECG record 119 (mV). 
Dashed line: the zero template, whose distance is the size of the reference; gray diamonds: the $\lambda \hsp{-2}{-2}{=} \infty$ straight-line template. 
For the EEG, the OLS reference equals the Average reference.
}
\label{fig:eeg_ecg}
\vspace{-3mm}
\end{figure*}

\begin{table}[t]
\caption{EEG ($\mu$V) and ECG (mV): mean distance to the OLS reference at $K \hsp{-2}{-2}{=} 5$ (20 selections). 
Bold: closest estimator; last two rows: baselines. 
}
\label{tab:summary}
\centering
\begin{tabular}{lcc}
\toprule
Method & Aud.\ EEG & ECG \\
\midrule
Average & 574.2 & \textbf{1.204} \\
OLS & 574.2 & \textbf{1.204} \\
GCV & 525.7 & 1.657 \\
L-curve & 543.8 & 3.374 \\
Emp.\ Bayes & 519.4 & 1.683 \\
SURE($\lambda$) & 557.9 & 1.237 \\
SURE(kernel) & \textbf{371.6} & 1.228 \\
\midrule
$\lambda \hsp{-2}{-2}{=} \infty$ (line) & 428.6 & 10.296 \\
Zero template & 287.3 & 10.734 \\
\bottomrule
\end{tabular}
\end{table}

\vspace{-3mm}
\subsection{P300 Brain--Computer Interface}
\label{sec:p300}

We use the BNCI2014-009 dataset~\cite{bnci2014009} from the MOABB benchmark suite~\cite{jayaram2018moabb}, 
which contains P300 speller data (described in~\cite{arico2014}).
We analyze all ten subjects with all 16~EEG channels at 128\,Hz, bandpass-filtered to 1--20\,Hz. 
Each subject has three recording blocks, filtered and resampled separately and joined in time order, 
with the block boundaries passed to the colored covariance, which then treats the noise in different blocks as uncorrelated (\Cref{rem:sigw_structure}).
Each block has 576 events (480 non-target, 96 target), with onsets 32 samples (250\,ms) apart. 
Template duration is 102~samples, from $-398$ to $+391$\,ms around the stimulus (nominally $\pm400$\,ms), so $P \hsp{-2}{-2}{=} 204$.
Every window lies at least 211 samples from both block edges. The full-data design has full column rank in every subject.
SURE($\lambda$) and SURE(kernel) use the colored-noise SURE with the AR(1) estimates of \Cref{sec:noise}. 

\textbf{Held-out evaluation.} Distance to a reference does not show if a template is useful. 
As the application endpoint we use single-trial target detection on held-out recording blocks: 
each of the three blocks is held out in turn, and every estimator is fitted on $K$ events per class selected from the other two blocks, using only those blocks. 
Each held-out window $\mathbf{Z}$ is scored by $s(\mathbf{Z}) \hsp{-2}{-2}{=} \|\mathbf{Z} - \hat{\mathbf{T}}_{\text{nt}}\|_F^2 - \|\mathbf{Z} - \hat{\mathbf{T}}_{\text{t}}\|_F^2$, 
where $\hat{\mathbf{T}}_{\text{t}}$ and $\hat{\mathbf{T}}_{\text{nt}}$ are the target and non-target templates, 
and the AUC is the Mann-Whitney probability that a target window scores above a non-target window, with ties counting one half~\cite{hanley1982auc}.

\textbf{Distance at $K \hsp{-2}{-2}{=} 5$.} \Cref{tab:p300_k5} summarizes the central comparison (7,200 fits). 
Among the subject-specific estimators, SURE(kernel) had the smallest mean distance under all three references, 
and it was closer than each of the other eight subject-specific estimators in all 10 subjects under all three references. 
In the primary family it was closer to the OLS reference than Empirical Bayes, SURE($\lambda$), 
and Average in all 10 subjects (Holm-adjusted $p \hsp{-2}{-2}{=} 0.006$ for each), 
and Molina's enhanced average~\cite{molina2024enhanced} was also farther in all 10 subjects (in its secondary family, Holm-adjusted $p \hsp{-2}{-2}{=} 0.004$). 
Selected per-subject results are given in the supplementary material. Empirical Bayes did not converge within 30 iterations in 61 of 200 fits; these fits are included.

\textbf{Trivial baselines.} To put these distances on a scale, \Cref{tab:p300_k5} and \Cref{fig:p300} also show the zero and population templates. 
At $K \hsp{-2}{-2}{=} 5$, the zero template is closer to the full-data OLS reference than each of the nine subject-specific estimators in every subject. 
The straight-line ($\lambda \hsp{-2}{-2}{=} \infty$) template is also closer on average than every subject-specific estimator, 
and the population template is the only template with a lower mean than zero. 
The reference is itself noisy: the median ratio of its model-based RMS noise to its norm is 0.89 across subjects (supplementary material). 

\begin{remark}[Why a trivial template can be closer]\label{rem:trivial}
The expected squared distance of any estimate to a noisy reference is the estimate's mean squared error plus the reference's noise, minus a term created by shared data; 
for the zero template it is the template's own squared size plus the same noise. 
An estimate is therefore farther from the reference than zero when its mean squared error, less the shared-data term, exceeds the template's squared size, 
which low signal-to-noise ratio can produce; since our metric averages norms rather than squared distances, this identity motivates the comparison rather than determining it. 
Model-based estimates, not true errors, suggest that the P300 data at $K \hsp{-2}{-2}{=} 5$ are in this regime, 
with the error of plain averaging about 5.5 times the template's squared size (supplementary material). 
There, shrinking toward zero lowers the error~\cite{hastie2009elements,hoerl1970ridge,james1961estimation,efron1973stein}, 
and the second-difference estimators cannot shrink below a straight line, which is itself farther from the reference than zero. 
Only at the largest tested $K \hsp{-2}{-2}{=} 50$ does a subject-specific cohort mean fall below the zero template, and only that of SURE(kernel) (\Cref{fig:p300}). 
We therefore compare estimators only relative to each other at matched $K$.
\end{remark}

\textbf{Dependence on $K$.} The cohort-mean distances of all subject-specific estimators decrease and their held-out AUCs increase with $K$, 
and the advantage of SURE(kernel) over Average shrinks (\Cref{fig:p300}).

\textbf{Held-out AUC.} At every $K$, every subject-specific estimator had a mean AUC above 0.5 in every subject (a descriptive comparison), but none was clearly better than averaging. 
At $K \hsp{-2}{-2}{=} 5$ the mean AUCs ranged from 0.597 (SURE(kernel)) to 0.607 (Empirical Bayes); 
the zero template gives 0.5 by construction and the $\lambda \hsp{-2}{-2}{=} \infty$ template 0.565. 
No estimator exceeded Average by more than 0.005 at any $K$ (largest mean paired difference: Empirical Bayes, $+0.0042$ at $K \hsp{-2}{-2}{=} 5$). 
The population template, which uses none of the subject's data, reached 0.687 and was above every $K \hsp{-2}{-2}{=} 5$ estimator in every subject. 

\begin{remark}[What the AUC measures]\label{rem:auc}
With the templates as vectors and $\langle\cdot,\cdot\rangle$ the Frobenius inner product, 
the score is $s(\mathbf{Z}) \hsp{-2}{-2}{=} 2\langle\mathbf{Z} - \mathbf{T}_{\text{c}}, \mathbf{T}_{\text{d}}\rangle$, 
with the difference template $\mathbf{T}_{\text{d}} \hsp{-2}{-2}{=} \hat{\mathbf{T}}_{\text{t}} - \hat{\mathbf{T}}_{\text{nt}}$ 
and the common template $\mathbf{T}_{\text{c}} \hsp{-2}{-2}{=} (\hat{\mathbf{T}}_{\text{t}} + \hat{\mathbf{T}}_{\text{nt}})/2$. 
The term $2\langle\mathbf{T}_{\text{c}}, \mathbf{T}_{\text{d}}\rangle$ is common to all windows, so the AUC depends only on the direction of $\mathbf{T}_{\text{d}}$: 
it ignores $\mathbf{T}_{\text{c}}$ and any positive rescaling of $\mathbf{T}_{\text{d}}$. 
In an exploratory decomposition (supplementary material), 55\% of SURE(kernel)'s reduction in summed squared distance to the full-data OLS reference, 
relative to Average at $K \hsp{-2}{-2}{=} 5$, was in $\mathbf{T}_{\text{c}}$ and 45\% in the scale of $\mathbf{T}_{\text{d}}$, with a negligible contribution from its direction: 
the gain lies in components that the AUC does not see. 
Template-implied event probabilities, which depend on this scale, were less overconfident for SURE(kernel) than for averaging at $K \hsp{-2}{-2}{=} 5$ 
(summary separation ratios 3.7 versus 8.9), with a mean information-gain improvement of 0.676 bits per event 
(positive in all 10 subjects) and a higher mean information gain at every tested $K$, at a similar AUC (supplementary material). 
Every subject-specific estimator nevertheless gave probabilities worse on average than the 1-in-6 prior at $K \hsp{-2}{-2}{=} 5$, 
and a scalar shrinkage of averaging gains slightly more information (0.081 bits per event more than SURE(kernel)) but lowers the AUC to 0.566. 
\end{remark}

\begin{table}[t]
\centering
\caption{P300, $K \hsp{-2}{-2}{=} 5$: mean distance ($\mu$V) to the three full-data references and held-out AUC, 
over 10 subjects (20 selections per subject; held-out: 3 folds $\times$ 20 selections). 
Lower distance and higher AUC are better. L-curve means exclude its 2 whole-recording and 15 held-out fits without a convex corner.}
\label{tab:p300_k5}
\vspace{-1mm}
\begin{tabular}{lcccc}
\toprule
 & \multicolumn{3}{c}{Distance to reference} & \\
\cmidrule(lr){2-4}
Estimator & OLS & EB & Average & AUC \\
\midrule
SURE(kernel) & 64.1 & 49.7 & 41.3 & 0.597 \\
Emp.\ Bayes & 77.0 & 64.6 & 55.3 & 0.607 \\
SURE($\lambda$) & 78.4 & 65.9 & 57.2 & 0.600 \\
GCV & 82.4 & 71.3 & 61.6 & 0.605 \\
L-curve & 91.6 & 82.1 & 72.7 & 0.603 \\
Average & 98.5 & 90.2 & 80.9 & 0.602 \\
OLS & 98.7 & 90.3 & 81.1 & 0.602 \\
Molina et al.~\cite{molina2024enhanced} & 112.4 & 105.2 & 96.5 & 0.599 \\
\midrule
$\lambda \hsp{-2}{-2}{=} \infty$ (line) & 60.8 & 45.4 & 40.9 & 0.565 \\
Zero template & 57.6 & 41.1 & 36.0 & 0.500 \\
Population & 56.8 & 42.1 & 37.5 & 0.687 \\
\bottomrule
\end{tabular}
\vspace{-3mm}
\end{table}

\begin{figure*}[t]
\centering
\includegraphics[width=\doublefigwidth]{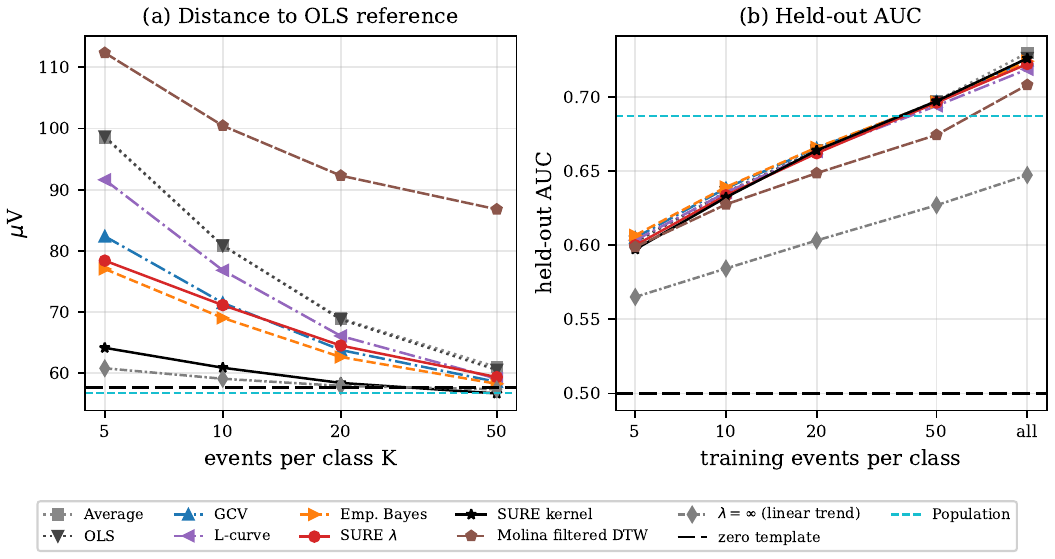}
\vspace{-2mm}
\caption{P300, 10 subjects: cohort means (of subject means) against the number of events per class $K$. 
(a)~Distance to the full-data OLS reference ($\mu$V). (b)~Held-out AUC (``all'': all training events). Dashed lines: the zero template (black; 
its AUC is 0.5) and the population template (cyan), which use none of the subject's selected events; 
gray diamonds: the $\lambda \hsp{-2}{-2}{=} \infty$ straight-line template. L-curve fits without a convex corner are excluded and counted: 
2 of 200 whole-recording fits and 15 of 600 held-out fits at $K \hsp{-2}{-2}{=} 5$, and 2 and 5 at $K \hsp{-2}{-2}{=} 10$. }
\label{fig:p300}
\vspace{-3mm}
\end{figure*}

\vspace{-3mm}
\subsection{ECG Beat Morphology}
\label{sec:ecg}

We use record 119 from the MIT-BIH Arrhythmia Database~\cite{moody2001mitbih,pollard2026physionet},
which contains Normal and premature ventricular contraction (PVC) beats recorded at 360\,Hz with a single ECG channel, bandpass-filtered to 0.5--40\,Hz.
The recording contains 1987 annotated beats (1543 Normal, 444 PVC). 
Template duration is 700\,ms ($\pm350$\,ms around the R-peak, 252~samples).
Since ECG baselines between adjacent beats may be contaminated by P-waves and T-waves from neighboring beats,
we use pairwise noise estimation for this dataset, which cancels the template contribution exactly when responses do not overlap.
\Cref{fig:eeg_ecg}(b) shows the distance to the full-data OLS reference. 
On this record, averaging and OLS were closest to the reference at $K \hsp{-2}{-2}{=} 5$ (1.204\,mV). 
SURE(kernel) (1.228\,mV) and SURE($\lambda$) (1.237\,mV) stayed within 3\% of averaging at every tested $K$ (5, 10, and 25), 
whereas GCV (1.657\,mV) and Empirical Bayes (1.683\,mV) were 38\% and 40\% farther than averaging at $K \hsp{-2}{-2}{=} 5$, and the L-curve was farthest (3.374\,mV). 
The order is the same under the EB reference (supplementary material). Every estimator was far closer to the reference than the zero template (10.73\,mV) at every $K$. 
In this comparison, no regularized estimator had a smaller mean distance to the OLS reference than averaging at the tested $K$, and the SURE estimators stayed close to it. \Cref{tab:summary} gives the distances at $K \hsp{-2}{-2}{=} 5$ for the auditory EEG and the ECG.

\vspace{-3mm}
\section{Discussion}
\label{sec:discussion}

\subsection{\texorpdfstring{Learned Regularization Structure Across Modalities}{Learned Regularization Structure Across Modalities}}
\label{sec:learned_structure}

\Cref{fig:kernels_cross} compares the normalized magnitude response of kernels learned by SURE(kernel) on each dataset.
The learned structures differ substantially across datasets. 
We summarize each kernel by its Nyquist-to-DC ratio: its gain at the fastest possible oscillation ($f \hsp{-2}{-2}{=} 0.5$ cycles per sample) 
divided by its gain at a constant offset ($f \hsp{-2}{-2}{=} 0$) in \Cref{fig:kernels_cross}. 
A large ratio means that the penalty acts almost only on rapid sample-to-sample variation, as the second difference $[1,-2,1]$ does (its gain at a constant offset is zero); 
a ratio near 1 means that it also shrinks the overall level of the template, as a ridge penalty does. 
The median ratio was 190 in the simulation and 118 on the ECG, where the learned kernels closely resemble scaled versions of the $[1,-2,1]$ operator; 
4.46 on the auditory EEG; and 1.75 on P300 (median over subjects of each subject's median ratio). 
The subject-median P300 responses in \Cref{fig:kernels_cross}(c) have ratios from 1.14 to 3.44, all above 1: 
the kernels penalize Nyquist more strongly than DC, with a nonzero DC gain. Individual kernels vary: 28 of the 200 have a ratio below 1.

Every learned kernel has a nonzero DC gain, so each also shrinks the overall level of the template, 
whereas the fixed $[1,-2,1]$ kernel has zero DC gain and a fixed response shape, which changing $\lambda$ only rescales. 
Learning the shape is not always best: in the ablation of the supplementary material, a plain ridge penalty was slightly closer to the references on the three P300 subjects tested, 
and on the ECG record, with the pairwise noise input of the main analysis, no regularized estimator had a smaller mean distance than averaging (\Cref{sec:ecg}). 

\begin{figure*}[t]
\centering
\includegraphics[width=\doublefigwidth]{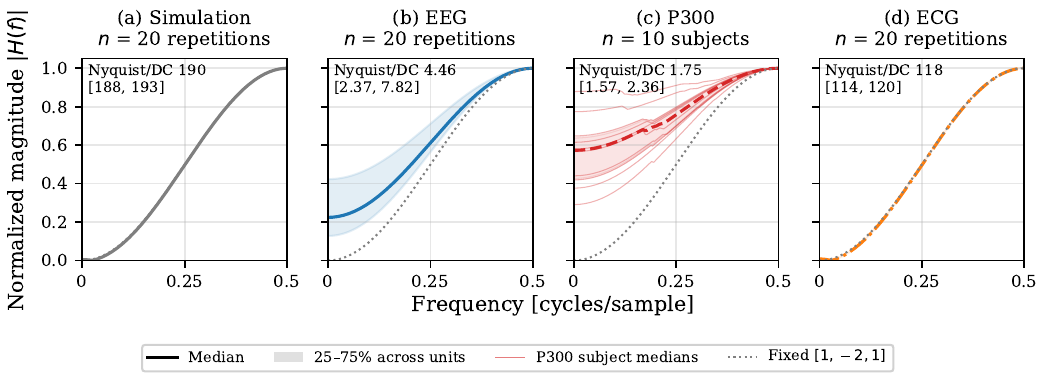}
\vspace{-2mm}
\caption{Normalized magnitude response $|H(f)|/\max_f |H(f)|$ ($f$ in cycles per sample) of the kernels learned by SURE(kernel). 
(a)~Simulation (noise sweep, $\sigma \hsp{-2}{-2}{=} 0.1$, 20 events per type), 
(b)~auditory EEG and (d)~ECG at $K \hsp{-2}{-2}{=} 5$: median and interquartile band over 20 realizations or selections. (c)~P300 at $K \hsp{-2}{-2}{=} 5$: 
median response of each subject's 20 kernels (thin lines), then the median and interquartile band across the 10 subjects. 
Dotted line: the fixed $[1,-2,1]$ kernel. Each panel gives the Nyquist-to-DC ratio, median [interquartile range] over kernels (P300: over the subjects' median ratios).
}
\label{fig:kernels_cross}
\vspace{-3mm}
\end{figure*}

\vspace{-3mm}
\subsection{Noise Estimation: A Design Choice}
\label{sec:discuss_noise}

SURE requires noise parameters as input, both a weakness (extra inputs required) and a strength (the user can adapt the noise model to the application). 
We used the AR(1)-colored model estimated from pre-stimulus baselines for the auditory EEG and P300 data 
and the pairwise estimate for the ECG and the white-noise simulations. 
Diagnostics on the baselines of all events and on the full-data residual series find autocorrelation structure beyond AR(1) on all three datasets. 
At $K \hsp{-2}{-2}{=} 5$, however, a richer model must be estimated from the same few short baselines: 
relative to AR(1), a third-order model estimated from them increased the geometric-mean distance of SURE(kernel) by 5\% to 26\% on three P300 subjects 
and the auditory EEG under all three references and by less than 1\% on the ECG, and the pairwise estimate increased it more (supplementary material). 
The choice thus trades model fidelity against estimation variance; a richer model may pay off when more events are available, which we did not test. 
On three P300 subjects, changing $\hat\varphi$ changed the subject-averaged geometric mean distance of SURE(kernel) by at most about 7\%, with no consistent direction, 
whereas scaling the noise-variance input by 1.25 lowered it and by 0.8 raised it in 18 of 18 subject, reference, and band cells (supplementary material). 
The pairwise estimate is useful when no clean baseline exists, but it overestimates the variance when responses overlap densely and the noise is weak, 
as in the low-noise simulations (\Cref{sec:sim_results}).

\vspace{-3mm}
\subsection{Computational Considerations}
\label{sec:discuss_computation}

On one BLAS thread, SURE(kernel) took 0.41\,s for the auditory EEG ($P \hsp{-2}{-2}{=} 480$, 59 channels), against 0.52\,s for Empirical Bayes; 
its time grew roughly as $P^2$, and the number of channels mattered little (supplementary material).
Because $P \hsp{-2}{-2}{=} ML$ does not depend on the number of channels, the cost of SURE(kernel) is governed by the template length and the number of event types.

\vspace{-3mm}
\subsection{Limitations}
\label{sec:discuss_limits}

Several limitations should be noted.
First, the AR(1) noise model is a first-order approximation: band-pass-filtered noise has oscillatory autocorrelation that AR(1) captures only approximately, 
and higher-order models, which need more parameters, may help when more events are available to estimate them (\Cref{sec:discuss_noise}).

Second, for real-data evaluation the reference template is estimated from all available events. No reference is ground truth: 
the full-data OLS reference carries noise of about the size of the template, and it shares the selected events' data with the low-$K$ estimate, a share that grows with $K$. 
We report three references and held-out AUC, which needs no reference. 
The held-out evaluation is within-subject, across recording blocks of one session, and covers one P300 paradigm with ten subjects.

Third, template non-stationarity (trial-to-trial variability) is not modeled. Our framework assumes a deterministic template, 
which is standard in the literature but may underestimate the effective noise variance in paradigms with substantial habituation or learning effects.

Finally, the method applies to temporal regularization only and is complementary to spatial filtering methods such as xDAWN~\cite{rivet2009xdawn}. 
Joint spatio-temporal optimization is a promising direction for future work.

\vspace{-3mm}
\section{Conclusion}
\label{sec:conclusion}

We presented a template estimation method that learns the Tikhonov penalty, both its structure and its magnitude, 
from the data by minimizing Stein's unbiased risk estimate over a convolution kernel, with a correction for temporally correlated noise. 
We compared it with established estimators in simulations and on auditory EEG, P300, and ECG recordings.

Among the compared estimators, the learned penalty was the most accurate in simulations at moderate and high noise and, with five events per class, 
the closest to the full-data references on the auditory EEG and in all ten P300 subjects. 
On the ECG record with pairwise noise estimation, where no regularized estimator beat averaging, it stayed within 3\% of averaging, 
while Empirical Bayes and generalized cross-validation were 38\% to 40\% farther at five events.

Three insights follow. First, the structure of the penalty matters, not only its magnitude: 
the learned kernels penalized almost only rapid variation on the ECG and in simulation but shrank the overall level of the template much more on the EEG datasets, 
which no scaling of a fixed difference operator can do. 
Second, more accurate templates need not improve template-matching AUC, which depends only on the direction of the target-minus-non-target template; 
they matter where the level and scale of the template are used, as for event probabilities, which were less overconfident than with averaging on P300. 
Third, with five events a zero template was closer on average to the reference than every subject-specific estimator on both EEG datasets, 
consistent with a low signal-to-noise regime (\Cref{rem:trivial}), so few-event evaluations should report trivial baselines and make only relative claims. 
This regime motivates shrinkage toward a population template, which gave the highest held-out AUC on P300 at five events; 
richer noise models for more events and spatio-temporal regularization are further directions. 

\section*{Acknowledgment}
The authors thank the Associate Editor and the anonymous reviewers for their constructive comments, which improved this paper.

\bibliographystyle{IEEEtran}
\vspace{-3mm}
\bibliography{references}

\clearpage
\includepdf[pages=-]{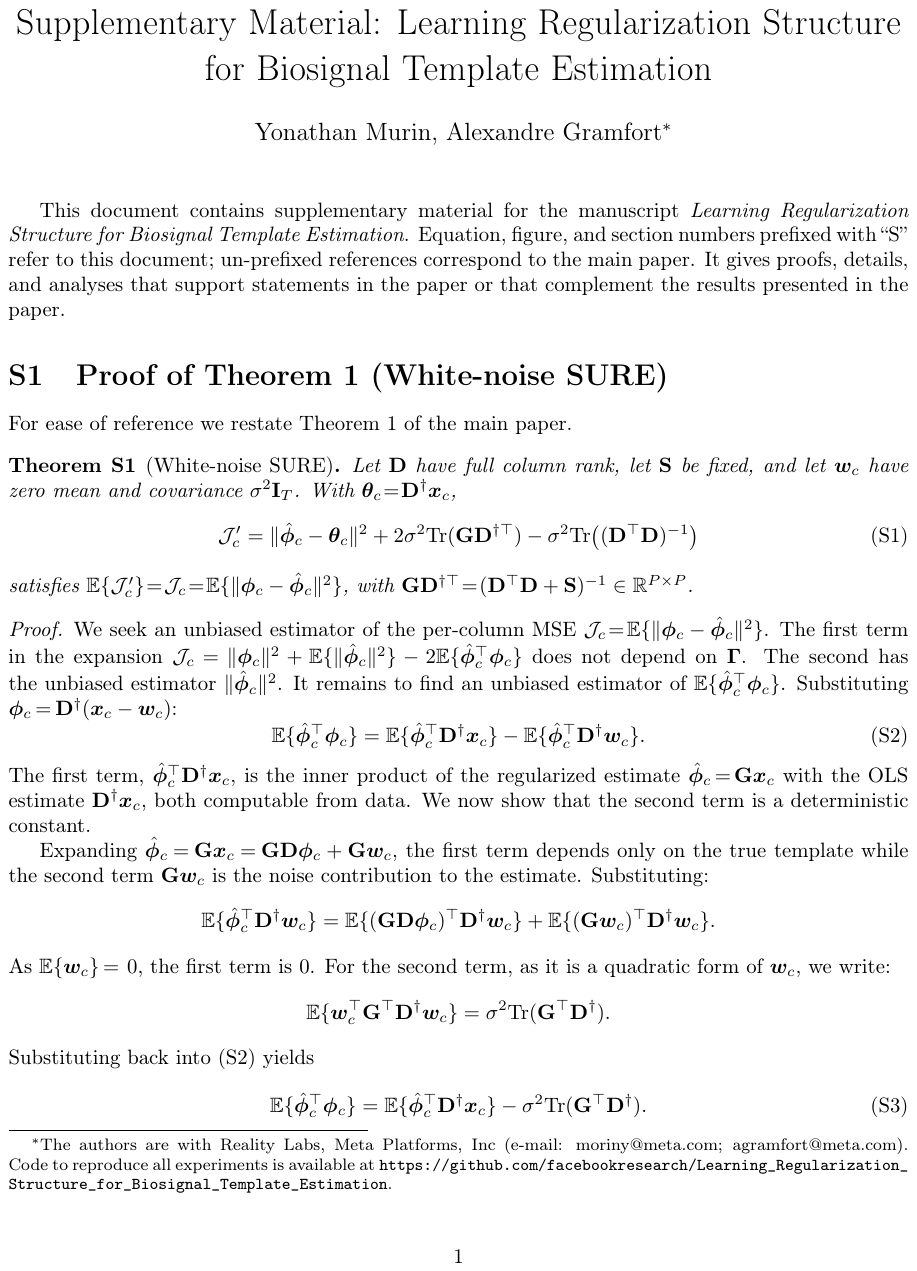}

\end{document}